\documentclass[conference]{IEEEtran}
\usepackage{amsfonts}
\usepackage{graphicx}
\usepackage{color}
\usepackage{amsmath,amsfonts,amssymb,amsthm,epsfig,epstopdf,url,array}
\usepackage{url,textcomp}
\usepackage{authblk}
\usepackage[T1]{fontenc}
\usepackage[utf8]{inputenc}
\usepackage{cite}
\usepackage[top=0.7in, bottom=0.7in, left=0.57in, right=0.57in]{geometry}

\newtheorem{theorem}{Theorem}

\begin{document}
\title{Performance Analysis of MIMO-NOMA Systems with Randomly Deployed Users}
\author[$\dag$]{Zheng~Shi}
\author[$\dag$]{Guanghua~Yang}
\author[$\S$]{Yaru~Fu}
\author[$\P$]{Hong~Wang}
\author[$\ddag$]{Shaodan~Ma}
\affil[$\dag$]{The School of Electrical and Information Engineering and Institute of Physical Internet, Jinan University, China}
\affil[$\S$]{Institute of Network Coding, The Chinese University of Hong Kong, Hong Kong}
\affil[$\P$]{School of Communication and Information Engineering, Nanjing University of Posts and Telecommunications, China}
\affil[$\ddag$]{Department of Electrical and Computer Engineering, University of Macau, Macao}
\maketitle
\begin{abstract}

This paper investigates the performance of Multiple-input multiple-output non-orthogonal multiple access (MIMO-NOMA) systems with randomly deployed users, where the randomly deployed NOMA users follow Poisson point process (PPP), the spatial correlation between MIMO channels are characterized by using Kronecker model, and the composite channel model is used to capture large-scale fading as well as small-scale fading. The spatial randomness of users' distribution, the spatial correlation among antennas and large-scale fading will severally impact the system performance, but they are seldom considered in prior literature for MIMO-NOMA systems, and the consideration of all these impact factors challenges the analysis. Based on zero-forcing (ZF) detection, the exact expressions for both the average outage probability and the average goodput are derived in closed-form. Moreover, the asymptotic analyses are conducted for both high signal-to-noise ratio (SNR) (/small cell radius) and low SNR (/large cell radius) to gain more insightful results. In particular, the diversity order is given by $\delta = {{N_r} - {M} + 1}$, the average outage probability of $k$-th nearest user to the base station follows a scaling law of $O\left({D^{\alpha \left( {{N_r} - {M} + 1} \right) + 2k}}\right)$, the average goodput scales as $O({D^{2}})$ and $O({D^{-2}})$ as $D \to 0$ and $D \to \infty$, respectively, where $N_r$, $M$, $\alpha$ and $D$ stand for the number of receive antennas, the total number of data streams, the path loss exponent and the cell radius, respectively. The analytical results are finally validated through the numerical analysis.

\end{abstract}
\IEEEpeerreviewmaketitle
\hyphenation{HARQ}
\section{Introduction}\label{sec:int}
%
%
%
%
%
%

The non-orthogonal multiple access (NOMA) has been envisioned to be a promising technique of 5G to achieve superior spectral efficiency to the conventional orthogonal multiple access (OMA), and the fundamental of NOMA is to multiplex multiple users in power domain, while they share the same frequency/time/code resources. Since multiple-input-multiple-out (MIMO) brings extra degrees of freedom, the integration of NOMA with MIMO has emerged as a new paradigm for further capacity enhancement and massive connectivity support\cite{dai2015non}.

Both optimal design and performance analysis for MIMO-NOMA system have attracted intensive research interest in the existing literature. Specifically, in \cite{sun2015ergodic}, a power allocation scheme was developed for the MIMO-NOMA system with a pair of users to maximize the ergodic capacity. In order to serve as many pairs of users as possible, a signal assignment based downlink and uplink MIMO-NOMA framework was proposed, where the signal assignment was particularly devised to remove inter-pair interference, and different power allocation strategies were then examined in \cite{ding2016general}. Bearing the idea of signal assignment in mind, it was further proved in \cite{liu2016capacity} that the ergodic capacity of MIMO-NOMA system outperforms that of MIMO-OMA system. Hence, the applications of NOMA to MIMO system in \cite{sun2015ergodic,ding2016general,liu2016capacity} were virtually developed for two typical NOMA users because the signal receptions for different pairs of users are processed independently. To accommodate more than two users, in \cite{wang2018novel}, the users within a cell were divided into multiple NOMA groups without a limit on the number of users in each group. By adopting group interference cancellation and minimum mean square error (MMSE) detector, a novel transceiver design was then proposed for MIMO-NOMA uplink transmission in \cite{wang2018novel}. Unfortunately, most of previous literature assumed perfect channel state information (CSI) at the transmitter, which would yield high signaling and computational overhead. Particularly in massive MIMO systems, the signal-to-noise (SNR) reporting and the computational complexity increase with the number of antennas, which will consume a significant proportion of resources \cite{ding2016design}. Therefore, NOMA schemes with limited CSI feedback are sometimes applied to MIMO systems \cite{ding2016design,ding2016application}. In \cite{ding2016mimo}, only statistical knowledge of CSI was assumed to be available to the base station for further reduction of the system overhead. Since the superiority of NOMA scheme considerably relies on the difference among users' channel conditions, the precoding and detection strategies were developed for MIMO-NOMA system to provide a new method to differ users' effective channel gains in \cite{ding2016mimo}.

Unfortunately, most prior works investigated the MIMO-NOMA systems by assuming a fixed topology and a fixed number of NOMA users, independence among MIMO channels, and only small-scale fading experienced by signals. Whereas the neglected factors including the spatial randomness of users' distribution, the spatial correlation among antennas, and the large-scale fading would inevitably impact the system performance, thus limit their wide applications in practice. To thoroughly examine these impact factors, a more general model for MIMO-NOMA systems is considered in this paper, in which the randomly deployed NOMA users follow a distribution of Poisson point process (PPP), the spatial correlation between MIMO channels are characterized by using Kronecker model, and the composite channel model is used to capture large-scale fading as well as small-scale fading. The introduction of all these factors makes the performance analysis extremely involved. By employing zero-forcing (ZF) detection at the users, this paper derives the exact expressions for both the average outage probability and the average goodput in closed-form. Based on the exact results, the asymptotic analysis is then conducted for both high SNR (/small cell radius) and low SNR (/large cell radius) to reveal more insights into the MIMO-NOMA system behavior. Finally, the validity of analytical results is confirmed via the numerical analysis.

The remainder of the paper is structured as follows. The mathematical model for the MIMO-NOMA system with randomly deployed users is developed in Section \ref{sec:sys_mod}. The exact and the asymptotic analyses are conducted for both the outage probability and the average goodput in Section \ref{sec:per_ana}. The simulation outcomes in Section \ref{sec:num} verify the analytical results along with some discussions. The paper is finally concluded with some important remarks in Section \ref{sec:con}.

\section{System Model}\label{sec:sys_mod}
In this paper, a downlink single-cell network is considered, where users are uniformly distributed in a disk with radius $D$ and the base station is located at the center. Specifically, the users within the disc are assumed to follow a PPP with intensity $\lambda$. Moreover, the base station is equipped with $N_t$ antennas and each user is equipped with $N_r$ antennas. We assume that the base station delivers $M (\le \min{(N_t,N_r)})$ independent data streams to multiple users at the same frequency/time/code via NOMA transmission. In each transmission, the number of users is assumed to be $K$. Then the received signal at the $k$-th nearest user to the base station is given by
\begin{equation}
{{\bf{y}}_k} = {\sqrt{\ell(d_k)}}{{\bf{H}}_{k}}{{\bf{V}}}{{\bf{s}}} + {{\bf{n}}_k},
\end{equation}
where ${{\bf{H}}_{k}} \in {\mathbb C}^{N_r\times N_t}$ is the channel response matrix from the base station to user $k$; ${\bf V}=({\bf v}_1,\cdots,{\bf v}_M) \in {\mathbb C}^{N_t \times M}$ is the transmit beamforming matrix at the base station, i.e., $\left\| {{{\bf{v}}_m}} \right\|=1$ for $m=1,\cdots,M$; ${{\bf{s}}} = ({s}_1,\cdots,{s}_M)$ is the superimposed symbol vector at base station, and the transmit symbol vector is drawn from the zero-mean complex symmetric circularly normal distribution with unit variance, i.e., ${{\bf{s}}} \sim {\cal C N}({\bf 0}, P{\bf I})$, and $P$ is the transmit power; ${{\bf{n_k}}}$ is the additive Gaussian white noise (AWGN) vector with variance $\sigma^2$, i.e., ${{\bf{n_k}}} \sim {\cal C N}({\bf 0}, \sigma^2{\bf I})$. A composite channel model is adopted to thoroughly examine impacts of both small-scale Rayleigh fading and large-scale path loss. In particular, the path loss is modelled by using Friis equation, such that $\ell(d_k)=\mathcal K {d_k}^{-\alpha}$ with path loss exponent $\alpha > 2$, where $d_k$ denotes the distance between the base station and user $k$, $\mathcal K $ represents the free space power received at the reference distance $d_0=1$m. 
Without loss of generality, the distances between the base station and users are sorted in an increasing order such that $d_1 \le d_2 \le \cdots \le d_K$.

For the sake of mathematical tractability in the sequel, a semi-correlated Rayleigh MIMO channel is assumed in this paper. More specifically, ${{\bf{H}}_{k}}$ is a complex circularly symmetric Gaussian matrix with zero mean and transmit-side covariance matrix ${\bf R}_T$. Hence, ${{\bf{H}}_{k}}$ can be written as
\begin{equation}
{{\bf{H}}_{k}} = {\bf H}_w {{\bf R}_T}^{1/2},
\end{equation}
where the elements of matrix ${\bf H}_w$ are i.i.d. complex Gaussian random variables with zero mean and variance ${\sigma_h}^2$, i.e., ${\rm vec}\left({\bf H}_k\right) \sim {\cal C N}({\bf 0}, {\sigma_h}^2 {{\bf R}_T}^{\rm T}\otimes {\bf I}_{{{N_R} }})$, where the notation $(\cdot)^{\rm T}$ refers to the transpose of a matrix and $\otimes$ is the Kronecker product. Besides, we assume that only statistical CSI is available to the base station, but the knowledge of CSI is perfectly known at users' side.

To simultaneously accommodate multiple users over the same frequency/time/code, NOMA scheme is applied herein. The information bearing vector of $M$ data streams is constructed by using superposition coding as follows
\begin{equation}
{\bf{s}} = \sqrt P \sum\limits_{k = 1}^K {\sqrt{\zeta  _k}{{\bf{x}}_k}} ,
\end{equation}
where ${{\bf{x}}_k}=(x_{k,1},\cdots,x_{k,M})^{\rm T}$ is the signal vector intended for user $k$, $\zeta _k$ is the power allocation coefficient, and $\sum\nolimits_{k = 1}^K {{\zeta _k}}  = 1$.

We assume that the received signal is reconstructed by employing a zero-frocing (ZF) detector. The received symbol stream is multiplied by ZF detector $\frac{1}{\sqrt{\ell(d_k)}}{\left( {{{\bf{H}}_k}{\bf{V}}} \right)^\dag } = \frac{1}{\sqrt{\ell(d_k)}}{\left( {{{\bf{V}}^{\rm{H}}}{{\bf{H}}_k}^{\rm{H}}{{\bf{H}}_k}{\bf{V}}} \right)^{ - 1}}{{\bf{V}}^{\rm{H}}}{{\bf{H}}_k}^{\rm{H}}$ so that
\begin{equation}
{\bf{\hat s}} = {\bf{s}} + \frac{1}{\sqrt{\ell(d_k)}}{\left( {{{\bf{V}}^{\rm{H}}}{{\bf{H}}_k}^{\rm{H}}{{\bf{H}}_k}{\bf{V}}} \right)^{ - 1}}{{\bf{V}}^{\rm{H}}}{{\bf{H}}_k}^{\rm{H}}{{\bf{n}}_k},
\end{equation}
where $({\bf X})^\dag$ refers to the pseudo-inverse of $\bf X$. With the linear ZF strategy, the decodings for $M$ data streams can be performed independently. This eventually yields a very low decoding complexity.

After ZF detection, the successive interference cancellation (SIC) as another key component of NOMA technique is adopted to separate superimposed signals sequentially. To simplify the design of power allocation coefficient and better utilize NOMA scheme, the decoding order at receivers is made according to the decreasing order of the distance, i.e., $d_K \ge  \cdots \ge  d_2 \ge d_1$. In particular, user $k$ will detect the message of user $i$ prior to decoding its own message ${\bf x}_k$ such that $i > k$, and then subtracting the signal ${\bf x}_{i}$ from its observation one by one. Meanwhile the signal ${\bf x}_{i}$ intended for user $i$ ($ < k$) will be treated as interference at user $k$. Therefore, for the $k$-th nearest user to the base station, the instantaneous received signal-to-interference-plus-noise ratio (SINR) to detect the signal $x_{i,m}$ that is extracted from data stream $m$ is given by
\begin{align}\label{eqn:sinr_x_im}
&{\gamma _{k \to i,m}} \notag \\
&\quad= \frac{{{\sigma _h}^2P{\zeta _i}}}{{{\sigma _h}^2P\sum\limits_{l = 1}^{i - 1} {{\zeta _l}}  + \frac{1}{\ell ({d_k})}{\sigma ^2}{{\left[ {{{\left( {{{\bf{H}}_k}{\bf{V}}} \right)}^\dag }{{\left( {{{\left( {{{\bf{H}}_k}{\bf{V}}} \right)}^\dag }} \right)}^{\rm{H}}}} \right]}_{mm}}}} \notag\\
&\quad=\frac{{{{\bar \gamma }}{\ell ({d_k})}{\zeta _i}}}{{{{\bar \gamma }}{\ell ({d_k})}\sum\limits_{l = 1}^{i - 1} {{\zeta _l}}  + {{\left[ {{{\left( {{{\left( {{{\bf{H}}_k}{\bf{V}}} \right)}^{\rm{H}}}{{\bf{H}}_k}{\bf{V}}} \right)}^{ - 1}}} \right]}_{mm}}}},\, i \ge k,
\end{align}
where ${{\bar \gamma }} = \frac{{P{\sigma _h}^2}}{{{\sigma ^2}}}$ denotes the average transmit SNR and $[{\bf X}]_{ij}$ is the $(i,j)$-th element of $\bf X$.

%

\section{Performance Analysis}\label{sec:per_ana}

\subsection{Exact Analysis}
Following the proposed MIMO-NOMA scheme, given the distance $d_k$ and the total number of users $K$, the probability of that user $k$ fails to decode its own message for data stream $m$ can be expressed as
\begin{multline}\label{eqn:out_pro_def}
p_{m,K,k}^{out} = \Pr \left[ {\bigcup\limits_{i = k}^K {{{\log }_2}\left( {1 + {\gamma _{k \to i,m}}} \right) < {R_{m,i}}} } \right]\\
= \Pr \left[ {{{\left[ {{{\left( {{{\left( {{{\bf{H}}_k}{\bf{V}}} \right)}^{\rm{H}}}{{\bf{H}}_k}{\bf{V}}} \right)}^{ - 1}}} \right]}_{mm}} > {{{\bar \gamma }_k{\theta _{m,k}}}{\ell ({d_k})}}} \right],
\end{multline}
where ${R_{m,i}}$ is the predefined transmission rate, ${\theta _{m,k}} = \min \left\{ {\left. {\frac{{{\zeta _i}}}{{{2^{{R_{m,i}}}} - 1}} - \sum\nolimits_{l = 1}^{i - 1} {{\zeta _l}} } \right|k \le i \le K} \right\}$ and $\bigcup{(\cdot)}$ stands for the union operation. It is worth mentioning that the condition of ${\theta _{m,k}}>0$ should be satisfied to successfully eliminate multiuser interference while performing SIC for NOMA. Thereupon, the requirement of $\left\{{\theta _{m,k}}>0,k\in[1,K]\right\}$ leads to $\left\{{R_{m,i}} < {\log _2}\left( {1 + \frac{{{\zeta _i}}}{{\sum\nolimits_{l = 1}^{i - 1} {{\zeta _l}} }}} \right),\,i\in[2,K]\right\}$.

By defining ${\bf{Z}} = {\left( {{{\bf{H}}_k}{\bf{V}}} \right)^{\rm{H}}}{{\bf{H}}_k}{\bf{V}}$, we find that $\bf{Z}$ is a complex Wishart matrix, i.e., ${\bf{Z}} \sim {\cal CW}_{{M}}\left( {{N_r},{\bf{\Sigma }}}\right)$, where ${\bf{\Sigma }} = {{\bf{V}}^{\rm{H}}}{{\bf{R}}_T}{\bf{V}}$. By using the concept of Schur complement ${\left[ {{{\bf{Z}}^{ - 1}}} \right]_{mm}} = \frac{{\det {{\bf{Z}}_{mm}}}}{{\det {\bf{Z}}}} = \frac{1}{{\det {\bf{Z}}_{mm}^{sc}}}$, the outage probability $p_{m,K,k}^{out}$ can be rewritten as \cite{gore2002transmit}
\begin{equation}\label{eqn:out_pro_re}
p_{m,K,k}^{out} = \Pr \left[ {\det {\bf{Z}}_{mm}^{sc} < \frac{1}{{{{{{\bar \gamma }}}\theta _{m,k}}{\ell ({d_k})}}}} \right].
\end{equation}

\cite{gore2002transmit} has proved that the Schur complement $\det {\bf{Z}}_{mm}^{sc}$ obeys Gamma distribution with a probability density function (PDF) of
\begin{equation}\label{eqn:schur_pdf}
{f_{\det {\bf{Z}}_{mm}^{sc}}}\left( x \right) = \frac{{{\beta _m}{e^{ - x{\beta _m}}}}}{{\left( {{N_r} - {M}} \right)!}}{\left( {{\beta _m}x} \right)^{{N_r} - {M}}},
\end{equation}
where $\beta_m ={\left[ {{{\bf{\Sigma }}^{ - 1}}} \right]_{mm}}$. Putting (\ref{eqn:schur_pdf}) into (\ref{eqn:out_pro_re}) then gives
\begin{align}\label{eqn:out_prob_gamma}
p_{m,K,k}^{out} 
&= \frac{{\gamma \left( {{N_r} - {M} + 1,\frac{{{\beta _m}}}{{{{\bar \gamma }}{\theta _{m,k}}{\ell ({d_k})}}}} \right)}}{{\Gamma \left( {{N_r} - {M} + 1} \right)}},
\end{align}
where $\gamma(a,b)$ is the lower incomplete gamma function.

Since the locations of $K$ users are uniformly distributed within the disc, the PDF of the Euclidean distance $d_k$ can be obtained by using order statistics as follows \cite{david1970order}
\begin{equation}\label{eqn:d_k_PDF}
{f_{{d_k}}}\left( x \right) = k\left( {\begin{array}{*{20}{c}}
K\\
k
\end{array}} \right){\left( {{F_d}\left( x \right)} \right)^{k - 1}}{\left( {1 - {F_d}\left( x \right)} \right)^{K - k}}{f_d}\left( x \right),
\end{equation}
where ${f_d}\left( x \right)$ and ${F_d}\left( x \right)$ are the PDF and the cumulative distribution function (CDF) of the distance $d$ from an arbitrary user to the base station, and are given by
\begin{equation}\label{eqn:dis_dis}
{f_d}\left( x \right) = \frac{{2x}}{{{D^2}}},{F_d}\left( x \right) = \frac{{{x^2}}}{{{D^2}}},\,x \le D.
\end{equation}

Then taking the expectation of (\ref{eqn:out_prob_gamma}) over the distribution of $d_k$ leads to
\begin{equation}\label{eqn:avg_pout}
\mathbb E{_{{d_k}}}\left( {p_{m,K,k}^{out}} \right) = \int\nolimits_0^D {p_{m,K,k}^{out}{f_{{d_k}}}\left( x \right)dx} \triangleq \bar p_{m,K,k}^{out}.
\end{equation}

Substituting (\ref{eqn:out_prob_gamma}) and (\ref{eqn:d_k_PDF}) into (\ref{eqn:avg_pout}), we have
\begin{align}\label{eqn:out_aver_dk_exp}
{\bar p_{m,K,k}^{out}} 
&= \frac{{2k}}{{\Gamma \left( {{N_r} - {M} + 1} \right)}}\left( {\begin{array}{*{20}{c}}
K\\
k
\end{array}} \right)\sum\limits_{j = 0}^{K - k} {\frac{{{{\left( { - 1} \right)}^j}}}{{{D^{2\left( {k + j} \right)}}}}\left( {\begin{array}{*{20}{c}}
{K - k}\\
j
\end{array}} \right) } \notag\\
&\times \int\limits_0^D {\gamma \left( {{N_r} - {M} + 1,\frac{{{\beta _m}{x^\alpha }}}{{{{\bar \gamma }}{\theta _{m,k}}{\cal K}}}} \right){x^{2\left( {k + j} \right) - 1}}dx}.
\end{align}

To proceed with the analysis, the lower incomplete Gamma function can be rewritten in terms of Meijer-G function by using \cite[eq.06.07.26.0006.01]{wolframe2010math},  i.e.,
\begin{equation}\label{eqn:incom_lower_gam}
\gamma \left( {a,x} \right)= G_{1,2}^{1,1}\left( {\left. {\begin{array}{*{20}{c}}
1\\
{a,0}
\end{array}} \right|x} \right),
\end{equation}
where $G_{p,q}^{m,n}(\cdot)$ refers to Meijer G-function. With (\ref{eqn:incom_lower_gam}), ${\bar p_{m,K,k}^{out}}$ can be further rewritten as
\begin{multline}\label{eqn:out_avg_mg}
{\bar p_{m,K,k}^{out}} = \frac{{2k}\left( {\begin{array}{*{20}{c}}
K\\
k
\end{array}} \right)}{{\Gamma \left( {{N_r} - {M} + 1} \right)}}\sum\limits_{j = 0}^{K - k} {\frac{{{{\left( { - 1} \right)}^j}}}{{{D^{2\left( {k + j} \right)}}}}\left( {\begin{array}{*{20}{c}}
{K - k}\\
j
\end{array}} \right)  }  \\
\times\int\limits_0^D {G_{1,2}^{1,1}\left( {\left. {\begin{array}{*{20}{c}}
1\\
{{N_r} - {M} + 1,0}
\end{array}} \right|\frac{{{\beta _m}{x^\alpha }}}{{{{\bar \gamma }}{\theta _{m,k}}{\cal K}}}} \right){x^{2\left( {k + j} \right) - 1}}dx}.
\end{multline}
By rewriting Meijer G-function in the form of Mellin-Barnes integral, we have
\begin{multline}\label{eqn:out_avg_mgmellin}
{\bar p_{m,K,k}^{out}} =
\frac{{2k}}{{\Gamma \left( {{N_r} - {M} + 1} \right)}}\left( {\begin{array}{*{20}{c}}
K\\
k
\end{array}} \right)\times\\
\sum\limits_{j = 0}^{K - k} {\frac{{{{\left( { - 1} \right)}^j}}}{{{D^{2\left( {k + j} \right)}}}}\left( {\begin{array}{*{20}{c}}
{K - k}\\
j
\end{array}} \right)  } \int\limits_0^D {{x^{2\left( {k + j} \right) - 1}}} \times\\
\frac{1}{{2\pi i}}\int\limits_{c - i\infty }^{c + i\infty } {\frac{{\Gamma \left( {{N_r} - {M} + 1 - s} \right)\Gamma \left( s \right)}}{{\Gamma \left( {1 + s} \right)}}{{\left( {\frac{{{\beta _m}{x^\alpha }}}{{{{\bar \gamma }}{\theta _{m,k}}K}}} \right)}^s}ds}dx .
\end{multline}
By switching the order of the two integrations, it follows that
\begin{multline}\label{eqn:out_avg_mg_sw}
{\bar p_{m,K,k}^{out}}
 = \frac{{2k}\left( {\begin{array}{*{20}{c}}
K\\
k
\end{array}} \right)}{{\Gamma \left( {{N_r} - {M} + 1} \right)}}\sum\limits_{j = 0}^{K - k} {\frac{{{{\left( { - 1} \right)}^j}}}{{{D^{2\left( {k + j} \right)}}}}\left( {\begin{array}{*{20}{c}}
{K - k}\\
j
\end{array}} \right)  } \\
\times\frac{1}{{2\pi i}}\int\limits_{c - i\infty }^{c + i\infty } {\frac{{\Gamma \left( {{N_r} - {M} + 1 - s} \right)\Gamma \left( s \right)}}{{\Gamma \left( {1 + s} \right)}}{{\left( {\frac{{{\beta _m}}}{{{{\bar \gamma }}{\theta _{m,k}}{\cal K}}}} \right)}^s}ds} \\
\times\int\limits_0^D {{x^{\alpha s + 2\left( {k + j} \right) - 1}}dx} ,\, \alpha s + 2\left( {k + j} \right) - 1 \ge 0.
\end{multline}
After some algebraic manipulations, (\ref{eqn:out_avg_mg_sw}) can be simplified as (\ref{eqn:out_avg_mg_alg}),
\begin{figure*}[!t]
\begin{align}\label{eqn:out_avg_mg_alg}
&{\bar p_{m,K,k}^{out}}
 = \frac{{2k}\left( {\begin{array}{*{20}{c}}
K\\
k
\end{array}} \right)}{{\Gamma \left( {{N_r} - {M} + 1} \right)}}\sum\limits_{j = 0}^{K - k} {{{\left( { - 1} \right)}^j}\left( {\begin{array}{*{20}{c}}
{K - k}\\
j
\end{array}} \right)  }\int\limits_{c - i\infty }^{c + i\infty } {\frac{{\Gamma \left( {{N_r} - {M} + 1 - s} \right)\Gamma \left( s \right)}}{{\Gamma \left( {1 + s} \right)}}\frac{{\Gamma \left( {\alpha s + 2\left( {k + j} \right)} \right)}}{{\Gamma \left( {\alpha s + 2\left( {k + j} \right) + 1} \right)}}{{\left( {\frac{{{\beta _m}{D^\alpha }}}{{{{\bar \gamma }}{\theta _{m,k}}{\cal K}}}} \right)}^s}ds}\\
\label{eqn:out_avg_mg_sw_H}
 &= \frac{{2k}}{{\Gamma \left( {{N_r} - {M} + 1} \right)}}\left( {\begin{array}{*{20}{c}}
K\\
k
\end{array}} \right)\sum\limits_{j = 0}^{K - k} {{{\left( { - 1} \right)}^j}\left( {\begin{array}{*{20}{c}}
{K - k}\\
j
\end{array}} \right)  }
H_{2,3}^{1,2}\left[ {\left. {\begin{array}{*{20}{c}}
{\left( {1,1} \right),\left( {1 - 2\left( {k + j} \right),\alpha } \right)}\\
{\left( {{N_r} - {M} + 1,1} \right),\left( { - 2\left( {k + j} \right),\alpha } \right),\left( {0,1} \right)}
\end{array}} \right|\frac{{{\beta _m}{D^\alpha }}}{{{{\bar \gamma }}{\theta _{m,k}}{\cal K}}}} \right].
\end{align}
\hrulefill
\end{figure*}
shown at the top of the next page. (\ref{eqn:out_avg_mg_alg}) can be further expressed in terms of Fox's H function as (\ref{eqn:out_avg_mg_sw_H}), shown at the top of the next page,
where the definition of Fox's H function can be found in \cite[eq.T.I.3]{ansari2017new}.

This paper considers the spatial randomness of the users, and the users are distributed as a PPP with intensity $\lambda$. Moreover, to suppress multiuser interference and reduce the hardware and computational complexity, the number of users in one NOMA group is limited up to $Q$ users, i.e., $K \le Q$. Accordingly, if there are more than $Q$ users in the disc, only $Q$ users are randomly selected from them to form the NOMA group. Hence, the number of users in the NOMA group is a random variable $\mathcal Q$ with probability distribution as
\begin{equation}\label{eqn:Q_dis}
\Pr \left( {{\cal Q} = K} \right) = \left\{ {\begin{array}{*{20}{c}}
{\frac{{{{\left( {\pi {D^2}\lambda } \right)}^K}}}{{K!}}{e^{ - \pi {D^2}\lambda }},}&{0 \le K < Q;}\\
{1 - {e^{ - \pi {D^2}\lambda }}\sum\limits_{q = 0}^{Q-1} {\frac{{{{\left( {\pi {D^2}\lambda } \right)}^q}}}{{q!}}} ,}&{K = Q.}
\end{array}} \right.
\end{equation}

Hence, the average outage probability of the $k$-th nearest user to the base station can be obtained by taking expectation over the probability distribution of $\cal Q$ as
\begin{multline}\label{eqn:out_fin_avg}
{{\tilde p}_{m,k}^{out}} = {\mathbb E_{\cal Q}}\left( {\bar p_{m,K,k}^{out}} \right) = {e^{ - \pi {D^2}\lambda }}\sum\limits_{K = k}^{Q-1} {\frac{{{{\left( {\pi {D^2}\lambda } \right)}^K}}}{{K!}}\bar p_{m,K,k}^{out}}  \\
+ \left( {1 - {e^{ - \pi {D^2}\lambda }}\sum\limits_{q = 0}^{Q-1} {\frac{{{{\left( {\pi {D^2}\lambda } \right)}^q}}}{{q!}}} } \right)\bar p_{m,Q,k}^{out},\,k \le K \le Q.
\end{multline}
It is noteworthy that the NOMA transmission reduces to a point-to-point communication when $k = 1$, the point-to-point communication in the circumstance can be viewed as a special case of NOMA transmission by setting $\alpha_1=1$ and all the other power allocation coefficients equal to zero.

With the outage expressions, the average goodput of the MIMO-NOMA system can be expressed as
\begin{equation}\label{eqn:avg_goodput}
\mathcal T_g = \sum\limits_{m = 1}^M {\sum\limits_{k = 1}^Q {\left( {1 - {e^{ - \pi {D^2}\lambda }}\sum\limits_{q = 0}^{k - 1} {\frac{{{{\left( {\pi {D^2}\lambda } \right)}^q}}}{{q!}}}  - {{\tilde p}_{m,k}^{out}}} \right){R_{k,m}}} }.
\end{equation}
where the goodput is frequently used to evaluate how many bits are successfully delivered per transmission \cite{shi2018energy}.

\subsection{Asymptotic Analysis}
Unfortunately, it is intractable to extract clear insights from the complicated expressions of the average outage probability and the average goodput due to the involvement of special functions. We thus proceed to carry out the asymptotic analysis for the outage probability and the goodput to gain more insightful results.
To start with, the asymptotic expression for ${\bar p_{m,K,k}^{out}}$ is derived first. (\ref{eqn:out_avg_mg_sw_H}) can be rewritten as
\begin{multline}\label{eqn:out_sim_r}
{\bar p_{m,K,k}^{out}} = \frac{{2k}}{{\Gamma \left( {{N_r} - {M} + 1} \right)}}\left( {\begin{array}{*{20}{c}}
K\\
k
\end{array}} \right)\sum\limits_{j = 0}^{K - k} {{\left( { - 1} \right)}^j}\left( {\begin{array}{*{20}{c}}
{K - k}\\
j
\end{array}} \right)\\
\times\frac{1}{{2\pi i}}\int\limits_{c - i\infty }^{c + i\infty } {\frac{{\Gamma \left( {{N_r} - {M} + 1 - s} \right)}}{{s\left( {\alpha s + 2\left( {k + j} \right)} \right)}}{{\left( {\frac{{{\beta _m}{D^\alpha }}}{{{{\bar \gamma }}{\theta _{m,k}}{\cal K}}}} \right)}^s}ds}.
\end{multline}
Since ${{\alpha s + 2\left( {k + j} \right) - 1}} > 1$, $c$ can be set as $c=0.5$. The asymptotic analysis can be split into two cases on the basis of whether ${\frac{{{\beta _m}{D^\alpha }}}{{{{\bar \gamma }}{\theta _{m,k}}{\cal K}}}} $ tends to $0$ or $\infty$.
\subsubsection{Asymptotic Analysis for High SNR or Small $D$}
Noticing that there are infinite number of simple poles on the right-hand-side of the vertical line from ${c - i\infty }$ to ${c + i\infty }$. Then applying Cauchy's residue theorem to (\ref{eqn:out_sim_r}), we arrive at
\begin{multline}\label{eqn:out_sim_re_exp}
{\bar p_{m,K,k}^{out}} = \frac{{ - 2k}}{{\Gamma \left( {{N_r} - {M} + 1} \right)}}\left( {\begin{array}{*{20}{c}}
K\\
k
\end{array}} \right)\sum\limits_{j = 0}^{K - k} {{\left( { - 1} \right)}^j}\left( {\begin{array}{*{20}{c}}
{K - k}\\
j
\end{array}} \right)\\
\sum\limits_{v = {N_r} - {M} + 1}^\infty  {{\rm{Res}}\left( {\frac{{\Gamma \left( {{N_r} - {M} + 1 - s} \right)}}{{s\left( {\alpha s + 2\left( {k + j} \right)} \right)}}{{\left( {\frac{{{\beta _m}{D^\alpha }}}{{{{\bar \gamma }}{\theta _{m,k}}{\cal K}}}} \right)}^s},v} \right)} ,
\end{multline}
where ${\rm{Res}}(\cdot,v)$ represents the residue at $s=v$. ${\bar p_{m,K,k}^{out}}$ can be expressed by calculating the residue as
\begin{align}\label{eqn:out_redi_th_alg}
&{\bar p_{m,K,k}^{out}} = \frac{{ - 2k}}{{\Gamma \left( {{N_r} - {M} + 1} \right)}}\left( {\begin{array}{*{20}{c}}
K\\
k
\end{array}} \right)\sum\limits_{j = 0}^{K - k} {{\left( { - 1} \right)}^j}\left( {\begin{array}{*{20}{c}}
{K - k}\\
j
\end{array}} \right)\times\notag\\
&\sum\limits_{v = {N_r} - {M} + 1}^\infty  {{{\left[ {\frac{{\left( {s - v} \right)\Gamma \left( {{N_r} - {M} + 1 - s} \right)}}{{s\left( {\alpha s + 2\left( {k + j} \right)} \right)}}{{\left( {\frac{{{\beta _m}{D^\alpha }}}{{{{\bar \gamma }}{\theta _{m,k}}{\cal K}}}} \right)}^s}} \right]}_{s = v}}}\notag\\
& = \frac{{2k}}{{\Gamma \left( {{N_r} - {M} + 1} \right)}}\left( {\begin{array}{*{20}{c}}
K\\
k
\end{array}} \right)\sum\limits_{j = 0}^{K - k} {{\left( { - 1} \right)}^j}\left( {\begin{array}{*{20}{c}}
{K - k}\\
j
\end{array}} \right)\times\notag\\
&{\sum\limits_{v = {N_r} - {M} + 1}^\infty  {{{\left[ {\frac{{\Gamma \left( {v + 1 - s} \right){{\left( {\frac{{{\beta _m}{D^\alpha }}}{{{{\bar \gamma }}{\theta _{m,k}}{\cal K}}}} \right)}^s}}}{{\prod\limits_{u = {N_r} - {M} + 1}^{v - 1} {\left( {u - s} \right)} s\left( {\alpha s + 2\left( {k + j} \right)} \right)}}} \right]}_{s = v}}} }  .
\end{align}
After some algebraic manipulations, ${\bar p_{m,K,k}^{out}}$ consequently can be expanded as
\begin{align}\label{eqn:out_fina_exp}
&{\bar p_{m,K,k}^{out}} = \frac{{2k}}{{\Gamma \left( {{N_r} - {M} + 1} \right)}}\left( {\begin{array}{*{20}{c}}
K\\
k
\end{array}} \right){\left( {\frac{{{\beta _m}{D^\alpha }}}{{{{\bar \gamma }}{\theta _{m,k}}{\cal K}}}} \right)^{{N_r} - {M} + 1}}\times\notag\\
&\quad \quad \quad\sum\limits_{j = 0}^{K - k} {{{\left( { - 1} \right)}^j}\left( {\begin{array}{*{20}{c}}
{K - k}\\
j
\end{array}} \right)} \times\notag\\
& \sum\limits_{\tau  = 0}^\infty  {\frac{{{{\left( { - \frac{{{\beta _m}{D^\alpha }}}{{{{\bar \gamma }}{\theta _{m,k}}{\cal K}}}} \right)}^\tau }}}{{\tau !\left( {\tau  + {N_r} - {M} + 1} \right)\left( {\alpha \left( {\tau  + {N_r} - {M} + 1} \right) + 2\left({k + j} \right)} \right)}}}.
\end{align}
With (\ref{eqn:out_fina_exp}), ${\bar p_{m,K,k}^{out}}$ can be rewritten as (\ref{eqn:out_asyp}), shown at the top of the next page,
\begin{figure*}[!t]
\normalsize
\begin{equation}\label{eqn:out_asyp}
{\bar p_{m,K,k}^{out}} = {{{\bar \gamma }}^{ - \left( {{N_r} - {M} + 1} \right)}}
\underbrace {\frac{{2k}}{{\Gamma \left( {{N_r} - {M} + 2} \right)}}\left( {\begin{array}{*{20}{c}}
K\\
k
\end{array}} \right){\left( {\frac{{{\beta _m}{D^\alpha }}}{{{\theta _{m,k}}K}}} \right)^{{N_r} - {M} + 1}}\sum\limits_{j = 0}^{K - k} {\frac{{{{\left( { - 1} \right)}^j}\left( {\begin{array}{*{20}{c}}
{K - k}\\
j
\end{array}} \right)}}{{\alpha \left( {{N_r} - {M} + 1} \right) + 2\left( {k + j} \right)}}} }_{{\vartheta _{m,K,k}}}
+o\left( {{{\bar \gamma }}^{ - \left( {{N_r} - {M} + 1} \right)}} \right).
\end{equation}
\hrulefill
\end{figure*}
where $o(\cdot)$ denotes the little-O notation, and we say $f(x) \in o(\phi(x))$ if $\lim {{f\left( x \right)} \mathord{\left/
 {\vphantom {{f\left( x \right)} {g\left( x \right)}}} \right.
 \kern-\nulldelimiterspace} {g\left( x \right)}} = 0$.

Putting (\ref{eqn:out_asyp}) into (\ref{eqn:out_fin_avg}), the resultant asymptotic expression of the average outage probability is given by
\begin{multline}\label{eqn:avg_out_asy}
{{\tilde p}_{m,k}^{out}} \approx {{\bar \gamma }}^{ - \left( {{N_r} - {M} + 1} \right)}\\
\times\left( \begin{array}{l}
{e^{ - \pi {D^2}\lambda }}\sum\limits_{K = k}^{Q-1} {\frac{{{{\left( {\pi {D^2}\lambda } \right)}^K}}}{{K!}}{\vartheta _{m,K,k}}} \\
 + \left( {1 - {e^{ - \pi {D^2}\lambda }}\sum\limits_{q = 0}^{Q-1} {\frac{{{{\left( {\pi {D^2}\lambda } \right)}^q}}}{{q!}}} } \right){\vartheta _{m,Q,k}}
\end{array} \right).
\end{multline}
From (\ref{eqn:avg_out_asy}), the insightful results regarding the diversity order and the outage scaling law against $D$ can be found in the following theorem, where the diversity order is commonly used to characterize the number of degrees of freedom in communication systems, and can be evaluated as $\delta = -\mathop {\lim }\limits_{\gamma  \to \infty } \frac{{\log \left({{\tilde p}_{m,k}^{out}}\right)}}{{\log \left( \gamma  \right)}}$.
\begin{theorem}\label{the:out_lowD}
The diversity order $\delta $ associated with the reception performance of the $k$-th nearest user to the base station is ${{N_r} - {M} + 1}$. In addition, as $D \to 0$, the average outage probability ${{\tilde p}_{m,k}^{out}}$ follows a scaling law of $O\left({D^{\alpha \left( {{N_r} - {M} + 1} \right) + 2k}}\right)$, where $O(\cdot)$ denotes the big-O notation.
\end{theorem}
\begin{proof}
Please see Appendix \ref{app:out_lowD}.
\end{proof}
From Theorem \ref{the:out_lowD}, it is readily found that the outage probability falls off as $D$ decreases, this is due to the fact that the reduction of transmission distance between the user and the base station yields a smaller path loss. However,
the small outage does not signify that the average goodput would increase with the decrease of $D$. Specifically, by utilizing the scaling law presented in Theorem \ref{the:out_lowD}, the asymptotic expression of the average goodput can be obtained as
\begin{multline}\label{eqn:goodput_asy_smallD}
{{\cal T}_g} = \sum\limits_{m = 1}^M {\sum\limits_{k = 1}^Q {\left( \begin{array}{l}
1 - {e^{ - \pi {D^2}\lambda }}\sum\limits_{q = 0}^{k - 1} {\frac{{{{\left( {\pi {D^2}\lambda } \right)}^q}}}{{q!}}} \\
 + O\left( {{D^{\alpha \left( {{N_r} - {M} + 1} \right) + 2k}}} \right)
\end{array} \right){R_{k,m}}} } \\
= {e^{ - \pi {D^2}\lambda }}\sum\limits_{m = 1}^M {\sum\limits_{k = 1}^Q {{R_{k,m}}\sum\limits_{q = k}^\infty  {\frac{{{{\left( {\pi {D^2}\lambda } \right)}^q}}}{{q!}}} } }  \\
+ O\left( {{D^{\alpha \left( {{N_r} - {M} + 1} \right) + 2k}}} \right)
 \propto {D^{2}},
\end{multline}
where $\propto$ stands for the operator of ``directly proportional to''. Accordingly, the average goodput ${{\cal T}_g}$ follows a scaling law of $O({D^{2}})$ as $D \to 0$. Although the path loss diminishes as $D$ decreases, the mean number of users, $\pi D^2\lambda$, decreases as well, and it dominates the goodput, which eventually causes the reduction of the goodput.
\subsubsection{Asymptotic Analysis for Low SNR or Large $D$}
In addition, it is of great significance to examine the trend of the outage probability with the increase of $D$ or with the decrease of SNR. Here we take the asymptotic analysis for large $D$ as an example, the asymptotic analysis for low SNR can be conducted in the same way. However, unlike Theorem \ref{the:out_lowD}, as $D$ approaches to infinity, i.e., $D \to \infty$, the average outage probability would behave differently.

By using \cite[Theorem 1.7]{kilbas2004h}, as ${\frac{{{D^\alpha }}}{{{{\bar \gamma }}}}} \to \infty$, the asymptotic expression of the outage probability can be obtained as (\ref{eqn:out_lowsnr_larD}), shown at the top of the next page,
\begin{figure*}[!t]
\normalsize
\begin{align}\label{eqn:out_lowsnr_larD}
&{\bar p_{m,K,k}^{out}} \notag\\
&= \frac{{2k}}{{\Gamma \left( {{N_r} - {M} + 1} \right)}}\left( {\begin{array}{*{20}{c}}
K\\
k
\end{array}} \right)\sum\limits_{j = 0}^{K - k} {{{\left( { - 1} \right)}^j}} \left( {\begin{array}{*{20}{c}}
{K - k}\\
j
\end{array}} \right)\left( \begin{array}{l}
{\rm{Res}}\left( {\frac{{\Gamma \left( {{N_r} - {M} + 1 - s} \right)}}{{s\left( {\alpha s + 2\left( {k + j} \right)} \right)}}{{\left( {\frac{{{\beta _m}{D^\alpha }}}{{{{\bar \gamma }}{\theta _{m,k}}K}}} \right)}^s},0} \right)\\
 + {\rm{Res}}\left( {\frac{{\Gamma \left( {{N_r} - {M} + 1 - s} \right)}}{{s\left( {\alpha s + 2\left( {k + j} \right)} \right)}}{{\left( {\frac{{{\beta _m}{D^\alpha }}}{{{{\bar \gamma }}{\theta _{m,k}}K}}} \right)}^s}, - \frac{{2\left( {k + j} \right)}}{\alpha }} \right) + o\left( {{{\left( {\frac{{{D^\alpha }}}{{{{\bar \gamma }}}}} \right)}^{ - \frac{{2\left( {k + j} \right)}}{\alpha }}}} \right)
\end{array} \right) \notag\\
&\approx \frac{{2k}}{{\Gamma \left( {{N_r} - {M} + 1} \right)}}\left( {\begin{array}{*{20}{c}}
K\\
k
\end{array}} \right)\sum\limits_{j = 0}^{K - k} {{{\left( { - 1} \right)}^j}} \left( {\begin{array}{*{20}{c}}
{K - k}\\
j
\end{array}} \right)\left( {\frac{{\Gamma \left( {{N_r} - {M} + 1} \right)}}{{2\left( {k + j} \right)}} - \frac{{\Gamma \left( {{N_r} - {M} + 1 + \frac{{2\left( {k + j} \right)}}{\alpha }} \right)}}{{2\left( {k + j} \right)}}{{\left( {\frac{{{\beta _m}{D^\alpha }}}{{{{\bar \gamma }}{\theta _{m,k}}K}}} \right)}^{ - \frac{{2\left( {k + j} \right)}}{\alpha }}}} \right) \notag\\
&=1 - \left( {\begin{array}{*{20}{c}}
K\\
k
\end{array}} \right)\frac{k}{{\Gamma \left( {{N_r} - {M} + 1} \right)}}\sum\limits_{j = 0}^{K - k} {{{\left( { - 1} \right)}^j}} \left( {\begin{array}{*{20}{c}}
{K - k}\\
j
\end{array}} \right)\frac{{\Gamma \left( {{N_r} - {M} + 1 + \frac{{2\left( {k + j} \right)}}{\alpha }} \right)}}{{k + j}}{\left( {\frac{{{\beta _m}{D^\alpha }}}{{{{\bar \gamma }}{\theta _{m,k}}K}}} \right)^{ - \frac{{2\left( {k + j} \right)}}{\alpha }}},\,D\to \infty,
\end{align}
\hrulefill
\vspace*{4pt}
\end{figure*}
where the last step holds by using the following identity, and the proof can be found in Appendix \ref{app:iden}.
\begin{equation}\label{eqn:identi}
\left( {\begin{array}{*{20}{c}}
K\\
k
\end{array}} \right)\sum\limits_{j = 0}^{K - k} {{{\left( { - 1} \right)}^j}} \left( {\begin{array}{*{20}{c}}
{K - k}\\
j
\end{array}} \right)\frac{k}{{k + j}} = 1.
\end{equation}

With (\ref{eqn:out_lowsnr_larD}), as $D$ increases, the average outage probability can be asymptotically expressed by noticing that $e^{{-\pi {D^2}\lambda }}\sum\nolimits_{q = 0}^{Q-1} {\frac{{{{\left( {\pi {D^2}\lambda } \right)}^q}}}{{q!}}}=o(D^{-2k})$ as
\begin{align}\label{eqn:out_fin_avglardD}
{{\tilde p}_{m,k}^{out}} = \bar p_{m,Q,k}^{out} + o(D^{-2k}).
\end{align}
Unlike the result in Theorem \ref{the:out_lowD}, the average outage probability approaches to 1 as $D \to \infty$. This is due to the increase of path loss when $D$ tends to $\infty$.

Moreover, plugging (\ref{eqn:out_fin_avglardD}) into (\ref{eqn:avg_goodput}) results in
\begin{align}\label{eqn:avg_goodput_large_D}
{{\cal T}_g} &= \sum\limits_{m = 1}^M {\sum\limits_{k = 1}^Q {\left( \begin{array}{l}
\left( {\begin{array}{*{20}{c}}
K\\
k
\end{array}} \right)\frac{{\Gamma \left( {{N_r} - {M} + 1 + \frac{{2k}}{\alpha }} \right)}}{{\Gamma \left( {{N_r} - {M} + 1} \right)}}\\
 \times {\left( {\frac{{{\beta _m} }}{{{{\bar \gamma }}{\theta _{m,k}}K}}} \right)^{ - \frac{{2k}}{\alpha }}}{D^{ - 2k}} + o\left( {{D^{ - 2k}}} \right)
\end{array} \right){R_{k,m}}} }\notag \\
 &\propto {D^{ - 2}}.
\end{align}
In contrast to the case of $D \to 0$, the goodput follows an opposite scaling law against $D$, that is, the goodput scales as $O\left({D^{- 2}}\right)$ as $D \to \infty$. However, this is not beyond our expectation, and can be explained as follows. The mean number of users and the average outage probability still exhibit two opposite impacts on the goodput. Unlike the case of $D \to 0$, the outage probability dominates the goodput as $D \to \infty$, and the increase of outage probability then deteriorates the goodput.

\section{Numerical Analysis and Discussions}\label{sec:num}
In this section, Monte Carlo simulations are provided to validate the theoretical analysis and the numerical results are also discussed. For illustration, the system parameters are set as follows. The small-scale fading channels are Rayleigh distributed with unit mean power, i.e.,  ${\sigma_h}^2=1$, the AWGN follows a complex Gaussian distribution with unit variance, i.e., $\sigma^2=1$, the path loss exponent is $\alpha=3$ and $\mathcal K=1$. All users are assumed to have the same transmission rate $R_{m,i}$ to ensure the fairness among users, i.e., $R_{m,i}=R$, and according to the prerequisites of NOMA transmission stated in (\ref{eqn:out_pro_def}), the power allocation coefficients are chosen such that ${\zeta _k} = \left( {1 - \sum\nolimits_{l = k + 1}^K {{\zeta _l}} } \right)\left( {1 - \varepsilon {2^{ - R}}} \right), k = 2,3,...,K$ and ${\zeta _1} = 1 - \sum\nolimits_{l = 2}^K {{\zeta _l}} $, where $\varepsilon$ is within $\left[ {0,1} \right]$, the transmit beamforming matrix $\bf V$ is set with ones on the principal diagonal and zeros elsewhere, the antenna correlation is modelled by using the common exponential correlation model as ${\bf R}_T = \left( [\rho^{|i-j|}]_{i,j}\right)$ \cite{li2017performance}. Unless otherwise specified, the system parameters $N_t$, $N_r$, $M$, $Q$, $\varepsilon$, $\rho$, $\bar \gamma$, $D$, $\lambda$ and $R$ are set to $2$, $3$, $2$, $3$, $0.5$, $0.5$, $60$dB, $30$m, $10^{-3}$m${}^{-2}$ and $2$bps/Hz, respectively.

Fig. \ref{fig:ver_out} plots the outage probability $\tilde{p}^{out}_{m=1,k=1}$ against the average transmit SNR $\bar \gamma$ under three different values of $D$. In the figure, the exact results, the asymptotic results and the simulation results are in perfect agreement, which confirms the correctness of the foregoing theoretical analysis. As can be observed in the figure, the curves for the outage probabilities under different values of $D$ gradually become parallel as $\bar \gamma$ increases. This is due to the fact that the slope of the curves against the SNR on a log-log scale is equivalent to the diversity order and the diversity order is independent of $D$, i.e., $\delta =N_r-M+1=2$. Thus it further justifies the validity of Theorem \ref{the:out_lowD}. 
\begin{figure}
  \centering
  \includegraphics[width=2.3in]{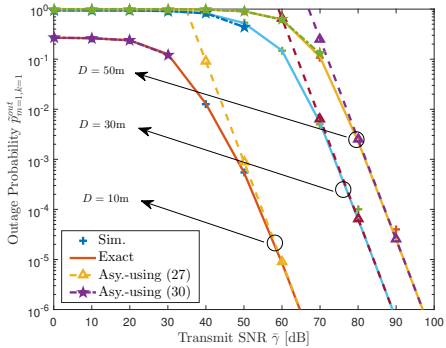}
  \caption{The outage probability $\tilde{p}^{out}_{m=1,k=1}$ versus the average transmit SNR $\bar \gamma$.}\label{fig:ver_out}
\end{figure}

The exact and the asymptotic analyses regarding the goodput $\mathcal T_g$ are also verified via the comparison with the simulation results, as shown in Fig. \ref{fig:ver_gp}. Apparently, the goodput is an increasing function of $\bar \gamma$, which coincides with the numerical results in Fig. \ref{fig:ver_gp}. It can be seen from the figure that the gooput converges to an upper bound as the average transmit SNR $\bar \gamma$ increases, and the upper bound of the goodput is undoubtedly determined by the cell radius $D$ from (\ref{eqn:avg_goodput}) when $\tilde{p}^{out}_{m,k}$ approaches to zero. Though it is roughly shown in Fig. \ref{fig:ver_gp} that the upper bound of the goodput is an increasing function of $D$, it does not mean that the goodput is also an increasing function of $D$. In fact, given $\bar \gamma$, from the preceding asymptotic analyses for both large $D$ and small $D$, the goodput would vanish no matter whether $D$ approaches to zero or infinity, which can be further justified in Fig. \ref{fig:D_gp}.

\begin{figure}
  \centering
  \includegraphics[width=2.3in]{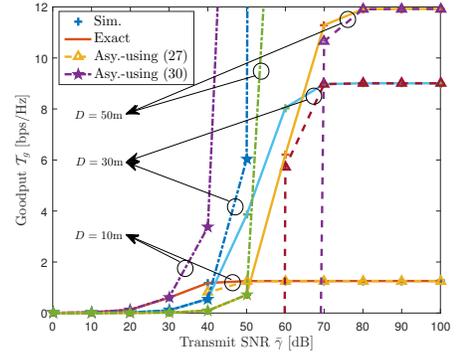}
  \caption{The goodput $\mathcal T_g$ versus the average transmit SNR $\bar \gamma$.}\label{fig:ver_gp}
\end{figure}

Fig. \ref{fig:D_gp} depicts the relationship between the goodput $\mathcal T_g$ and the cell radius $D$, in which the asymptotic results obtained by using (\ref{eqn:out_asyp}) and (\ref{eqn:out_lowsnr_larD}) perfectly match the exact results under small $D$ and large $D$, respectively. In both cases, the goodput converges to zero. Evidently there exists a maximal value of the goodput that can be reached through optimizing the cell radius $D$. For example, the optimal cell radius $D$ is around $35$m in Fig. \ref{fig:ver_gp}.
%

\begin{figure}
  \centering
  \includegraphics[width=2.3in]{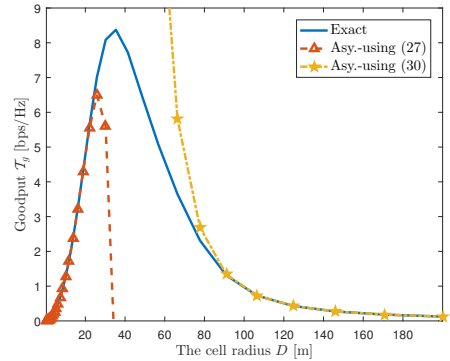}
  \caption{The goodput $\mathcal T_g$ versus the cell radius $D$.}\label{fig:D_gp}
\end{figure}

Moreover, Fig. \ref{fig:corr_gp} shows the impact of the antenna correlation on the goodput, where the strength of the spatial correlation is characterized by $\rho$. As shown in the figure, the spatial correlation negatively influences the goodput, and it can be explained as follows. It is not hard to find from both (\ref{eqn:out_asyp}) and (\ref{eqn:out_lowsnr_larD}) that the outage probability is an increasing function of $\beta_m$. For the exponential correlation model, $\beta_m$ is either $\frac{1}{1-\rho^2}$ or $\frac{1-\rho^4}{(1-\rho^2)^2}$. Clearly, $\beta_m$ is an increasing function of $\rho$. Thus the outage probability increases with $\rho$, which consequently leads to the reduction of the goodput.
\begin{figure}
  \centering
  \includegraphics[width=2.3in]{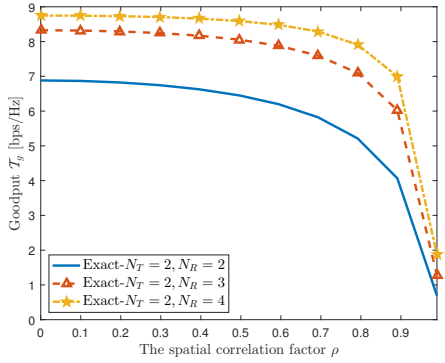}
    \caption{The goodput $\mathcal T_g$ versus the spatial correlation factor $\rho$.}\label{fig:corr_gp}
\end{figure}
\section{Conclusions}\label{sec:con}


This paper has thoroughly investigated the performance of MIMO-NOMA systems by assuming the randomly deployed NOMA users, the spatially correlated channels and the composite channel model. Based on ZF detection, the exact expressions for both the average outage probability and the average goodput have been derived in closed-form. The tractable exact results have enabled the asymptotic analyses with clear insights. More specifically, it has been proved that the diversity order is given by $\delta = {{N_r} - {M} + 1}$, the average outage probability of $k$-th nearest user to the base station follows a scaling law of $O\left({D^{\alpha \left( {{N_r} - {M} + 1} \right) + 2k}}\right)$, the average goodput scales as $O({D^{2}})$ and $O({D^{-2}})$ as $D \to 0$ and $D \to \infty$, respectively. The trivially simple expressions would greatly facilitate the future optimal system design and also have a wide range of other potential applications.

\section{Acknowledgements}
This work was supported in part by National Natural Science Foundation of China under grant 61601524, in part by the Macau Science and Technology Development Fund under grants
091/2015/A3 and 020/2015/AMJ, and in part by the Research Committee of University of Macau under grants MYRG2014-00146-FST and MYRG2016-00146-FST.
\appendices
\section{Proof of Theorem \ref{the:out_lowD}}\label{app:out_lowD}
It is clear from (\ref{eqn:avg_out_asy}) that the diversity order is $\delta = {{N_r} - {M} + 1}$. Moreover, in order to demonstrate the outage scaling law, (\ref{eqn:avg_out_asy}) can be rewritten as
\begin{multline}\label{eqn:out_rewr}
{{\tilde p}_{m,k}^{out}} = {{\bar \gamma }}^{ - \left( {{N_r} - {M} + 1} \right)}{e^{ - \pi {D^2}\lambda }} \\
 \times \left( {\sum\limits_{K = k}^{Q-1} {\frac{{{{\left( {\pi {D^2}\lambda } \right)}^K}}}{{K!}}{\vartheta _{m,K,k}}}  + {\vartheta _{m,Q,k}}\sum\limits_{q = Q }^\infty {\frac{{{{\left( {\pi {D^2}\lambda } \right)}^q}}}{{q!}}} } \right).
\end{multline}
By using the definition of ${\vartheta _{m,K,k}}$ in (\ref{eqn:out_asyp}), the outage scaling law with respect to $D$ directly follows.
\section{Proof of (\ref{eqn:identi})}\label{app:iden}
To prove (\ref{eqn:identi}), we first define the function $\phi\left( x \right) = k{x^k}\sum\limits_{j = 0}^{K - k} {{x^j}} \left( {\begin{array}{*{20}{c}}
{K - k}\\
j
\end{array}} \right)\frac{1}{{k + j}}$. Taking the first derivative of $\phi\left( x \right)$ with respect to $x$, it follows that
\begin{align}\label{eqn:iden_der}
\phi'\left( x \right)
&= k{x^{k - 1}}{\left( {1 + x} \right)^{K - k}}.
\end{align}
Accordingly, $\phi(-1)$ can be rewritten as
\begin{align}\label{eqn:phi_rew}
\phi\left( -1 \right) &= \int\nolimits_0^{ - 1} {\phi'\left( x \right)dx}  + \phi\left( 0 \right)  
=k{\left( { - 1} \right)^k}{\rm{B}}\left( {k,K - k + 1} \right),
\end{align}
where ${\rm B}(a,b)=\frac{\Gamma(a)\Gamma(b)}{\Gamma(a+b)}$ denotes Beta function. Comparing the definition of $\phi\left( x \right)$ with (\ref{eqn:phi_rew}) yields (\ref{eqn:identi}).
\bibliographystyle{ieeetran}
\bibliography{mimo_zf_noma}

\end{document}